\documentclass[runningheads]{llncs}

\usepackage{amsmath}
\usepackage{amssymb}
\usepackage{proof}
\usepackage{comment}
\usepackage{mathdots}
\usepackage{mycommands}

\usepackage{graphicx}
\begin{document}
\title{Decidable Exponentials in Nonassociative Noncommutative Linear Logic}
%
%
\author{Eben Blaisdell\inst{1}}
\authorrunning{E. Blaisdell}
%
\institute{Department of Mathematics, University of Pennsylvania
\email{ebenb@sas.upenn.edu}}
\maketitle              
\begin{abstract}
The use of exponentials in linear logic greatly enhances its expressive power.  In this paper we focus on nonassociative noncommutative multiplicative linear logic, and systematically explore modal axioms $\K$, $\T$, and $4$ as well as the structural rules of contraction and weakening.  We give sequent systems for each subset of these axioms; these enjoy cut elimination and have analogues in more structural logics.  We then appeal to work of Buli\'nska \cite{bulinska2009} extending work of Buszkowski \cite{buszkowski2005} to show that several of these logics are $\mathsf{PTIME}$ decidable and generate context free languages as categorial grammars.  This contrasts associative systems where similar logics are known to generate all recursively enumerable languages, and are thus in particular undecidable \cite{kanovich2020soft}

\end{abstract}

\section{Introduction}
Lambek introduced a noncommutative calculus \cite{Lambek1958} for the analysis of natural language syntax.  Three years later, Lambek, lead by linguistic examples, conceived a further restricted nonassociative calculus \cite{Lambek1961OnTC}.  In these frameworks, grammaticality of sentences corresponds with the provability of particular sequents.  Both of these calculi generate context free languages \cite{Pentus1993LambekGA,Kandulski1988TheEO}, and have been extended in several orthogonal ways by connectives that increase expressivity \cite{moortgat2007,kanovich2020soft,moortgat1997categorial}.

The original exponential of linear logic \cite{Girard:1987uq}, viewed as a modality acts as $\Sf$, i.e. it validates the following modal axioms:

\[
\K:!(A\to B)\to !A\to !B
\qquad
\T:!A\to A
\qquad
\4:!A\to !!A
\]

However, the choice of $\Sf$ is not a direct consequence of the view of bang as arbitrary reuse of a formula; exponentials satisfying different modal axioms have been widely studied \cite{girard1998,lafont2004}.

In linear logic \cite{Girard:1987uq}, the exponential licenses contraction and weakening, which can be axiomatized as:

\[
\C:!A\to !A\tensor !A
\qquad
\W:!A\to \one
\]

In linear logic, these allow a formula to be processed zero or more times as needed.  In noncommutative or nonassociative systems, these allow for arbitrary use of a formula locally.


While in intuitionistic linear logic these axioms are easily encoded in following the sequent rules.

\[
\infer[\C]{\Gamma,!A\seq C}{\Gamma,!A,!A\seq C}
\qquad
\infer[\W]{\Gamma,!A\seq C}{\Gamma\seq C}
\]

It is nontrivial to implement these rules well in nonassociative and noncommutative systems proving $\K$, as we will note later.

\section{The Systems}
This work grew out of work on \cite{blaisdell22multimodal}, and features similar formulas, rules, and systems.

\begin{definition}[Formulas]
Let a multiplicative exponential formula be the closure of the propositional variables $p_i$ for $i\in\mathbb{N}$ and the unit $\one$, under the unary symbol $!$ and the binary symbols $\tensor,\to,$ and $\la$.  A multiplicative formula is a multiplicative exponential formula with no occurrence of $!$.

Let $\F$ be the set of all multiplicative formulas, and $!\F$ be the set of all multiplicative exponential formulas.
\end{definition}

\begin{definition}[Structured Sequents]
A structure is a binary tree of formulas.  Put differently, it is defined by the grammar

\[
\Gamma \quad ::= \quad (\Gamma,\Gamma) \quad | \quad !\F.
\]

\noindent allowing additionally for an empty structure.

A structured sequent is written $\Gamma\seq C$ where $\Gamma$ is a structure and $C$ is a formula.
\end{definition}

\begin{figure}[t]
{\sc Propositional rules}
\[
\infer[\tensor L]{\Rx{A\tensor B}\seq C}{\Rx{(A, B)}\seq C}
\qquad
\infer[\tensor R]{(\Gamma,\Delta) \seq A\tensor B}{\Gamma\seq A &
\Delta\seq B}
\qquad
\infer[\to L]{\Rx{(\Delta,A\to B)} \seq C}{\Delta\seq A &
\Rx{B}\seq C}
\]
\[
\infer[\to R]{\Gamma\seq A\to B}{(A,\Gamma)\seq B}
\qquad
\infer[\la L]{\Rx{(B\la A,\Delta)} \seq C}{\Delta\seq A &
\Rx{B}\seq C}
\qquad
\infer[\la R]{\Gamma\seq B\la A}{(\Gamma,A)\seq B}
\]
\[
\infer[\one L]{\Rx{\one} \seq A}{\Rx{} \seq A}
\qquad
\infer[\one R]{ \seq \one}{}
\]
{\sc Modal Rules}
\[
\infer[!L]{\Rx{!A}\seq C}{\Rx{A}\seq C}
\qquad
\infer[!R]{!A\seq !C}{A\seq C}
\qquad
\infer[!R\K]{!\Gamma\seq !C}{\Gamma\seq C}
\]
\[
\infer[!R\4]{!A\seq !C}{!A\seq C}
\qquad
\infer[!R\K 4]{!\Gamma\seq !C}{!^*\Gamma\seq C}
\]
{\sc Structural Rules}
\[
\infer[\C]{\Rx{!A}\seq C}{\Rx{(!A,!A)}\seq C}
\qquad
\infer[\CK]{\Rx{!\Delta}\seq C}{\Rx{(!\Delta,!\Delta)}\seq C}
\qquad
\infer[\W]{\Rx{!A}\seq C}{\Rx{}\seq C}
\]
{\sc Initial and cut rules}
\[
\infer[\init]{A\seq A}{} 
\qquad
\infer[\cut]{\Rx{\Delta}\seq C}{\Delta\seq A & \Rx{A}\seq C}
\]
\caption{All rules for the set of systems defined.}\label{fig:rules}
\end{figure}

\begin{definition}[$\MacLL$]
Let \emph{multiplicative nonassociative noncommutative linear logic}, or $\MacLL$, be the logic with the rules $(\init),(\tensor L),(\tensor R),(\to L),(\to R),(\la L),$ and $(\la R)$.
\end{definition}

\begin{definition}[$\SMMacLL$]
For each set $S\subseteq\{c,w\}$ of labels corresponding to structural rules, and set $M\subseteq\{k,t,4\}$ of labels corresponding to modal axioms, we define $\SMMacLL$ to be the logic including

\begin{itemize}
    \item the rules of $\MacLL$,
    \item $(!R)$ if $M-t=\emptyset$ or $\{4\}$,
    \item $(!R\K)$ if $M-t=\{k\}$
    \item $(!R\4)$ if $M-t=\{4\}$,
    \item $(!R\K\4)$ if $M-t=\{k,4\}$,
    \item $(!L)$ if $t\in M$
    \item $(\C)$ if $c\in S$ and $k\not\in M$
    \item $(\C\K)$ if $c\in S$ and $k\in M$
    \item $(\W)$ if $w\in S$
\end{itemize}
\end{definition}

In $(!R\K\4)$, by $!^*\Gamma$, we mean $\Gamma$ with some subset of the formulas wrapped in bang $(!)$.

Note the peculiarity of the $(\C\K)$ rule.  Such a rule is needed because a system proving $\K$ and the single-formula $(\C)$ would not admit cut.  For example, $(!a\tensor !b)\seq (a\tensor b)\tensor (a\tensor b)$ is provable with

\[
\infer[\cut]{!a\tensor !b\seq (a\tensor b)\tensor (a\tensor b)}{
    \infer[\tensor L]{!a\tensor !b\seq !(a\tensor b)}{
    \infer[!R\K]{(!a,!b)\seq !(a\tensor b)}{
    \infer[\tensor R]{(a,b)\seq a\tensor b}{
        \infer[\init]{a\seq a}{} &
        \infer[\init]{b\seq b}{}
    }}} &
    \infer[\C]{!(a\tensor b)\seq(a\tensor b)\tensor (a\tensor b)}{
    \infer[\tensor R]{(!(a\tensor b),!(a\tensor b))\seq(a\tensor b)\tensor (a\tensor b)}{
        \infer[!L]{!(a\tensor b)\seq a\tensor b}{
            \infer[\init]{a\tensor b\seq a\tensor b}{}
        } &
        \infer[!L]{!(a\tensor b)\seq a\tensor b}{
            \infer[\init]{a\tensor b\seq a\tensor b}{}
        }
    }
    }
}
\]

\noindent but is not provable by $(\C)$ alone without $(\cut)$.

We will see that $(\CK)$ eliminates $(\cut)$.  However, it is slightly stronger than the axiom $\C$.  For example, in general $\C$ does not imply $!a\tensor !b\seq(!a\tensor !b)\tensor(!a\tensor !b)$, but this is provable with $(\CK)$.  However, in systems proving $\T$ and $\4$, we have that $(\CK)$ follows from $(\C)$ as follows:

\[
\infer[\cut]{\Rx{!\Delta}\seq C}{
    \infer[!R\K\4]{!\Delta\seq !\bigotimes!\Delta}{
    \infer=[\tensor R,\init]{!\Delta\seq \bigotimes!\Delta}{}
    } &
    \infer[\C]{\Rx{!\bigotimes!\Delta}\seq C}{
    \infer=[!L]{\Rx{(!\bigotimes!\Delta,!\bigotimes!\Delta)}\seq C}{
    \infer=[\tensor L]{\Rx{(\bigotimes!\Delta,\bigotimes!\Delta)}\seq C}{
    \Rx{(!\Delta,!\Delta)}\seq C
    }
    }}
}
\]

\section{Cut Admissibility}
The subformula property is vital to our proofs of decidability and context-freeness, so we pay close attention to the admissibility of $(\cut)$ in all of the logics defined above.

\begin{lemma}
If $\Delta\seq A$ and $\Rx{A}\seq C$ have (cut-free) proofs in $\SMMacLL$, then $\Rx{\Delta}\seq C$ has a (cut-free) proof in $\SMMacLL$.
\end{lemma}
\begin{proof}
We prove this by triple induction on

\begin{itemize}
    \item $\kappa$, the complexity of $A$, and
    \item $\delta$, the number of rule applications in the proofs of $\Delta\seq A$ and $\Rx{A}\seq C$.
\end{itemize}

We proceed casewise by the final rules in the proofs of $\Delta\seq A$ and $\Rx{A}\seq C$ if the top level connective of $A$ is not the exponential, and then handle the exponential case more carefully.

\subsubsection{Axioms}

If either proof is $(\init)$ alone, then the other proof suffices directly.

\[
\infer[\cut]{\Delta\seq A}{
    \Delta\seq A &
    \infer[\init]{A\seq A}{}
}
\qquad
\infer[\cut]{\Rx{A}\seq C}{
    \infer[\init]{A\seq A}{} &
    \Rx{A}\seq C
}
\]

\subsubsection{Left Nonprincipal}
If a left rule $\mathcal{R}$ is the final rule in the proof of $\Delta\seq A$, then we have the following general form, potentially with an extra assumption $\mathcal{A}$.

\[
\infer[\cut]{\Rx{\Delta}\seq C}{
    \infer[\mathcal{R}]{\Delta\seq A}{
        \mathcal{A}^? & 
        \Delta'\seq A
    } &
    \Rx{A}\seq C
}
\]

This can be rewritten to the following.

\[
\infer[\mathcal{R}]{\Rx{\Delta}\seq C}{
    \mathcal{A}^? &
    \infer[\cut]{\Rx{\Delta'}\seq C}{
        \Delta'\seq A &
        \Rx{A}\seq C
    }
}
\]

Here the formula complexity is maintained, but the number of rule applications is one fewer.

\subsubsection{Right Nonprincipal}
The case of nonprincipal rules in $\Rx{A}\seq C$ requires more granularity.

For a rule $\mathcal{R}$ that is $(\tensor L)$, $(!L)$, $(\C)$, $(\C\K)$, $(\W)$, or $(\one L)$ applied nonprincipally, or $(\to L)$ or $(\ra L)$ applied independently of $A$, potentially with another assumption $\mathcal{A}$, we have the following setup.

\[
\infer[\cut]{\Rx{\Delta}\seq C}{
    \Delta\seq A &
    \infer[\mathcal{R}]{\Rx{A}\seq C}{
        \mathcal{A}^? &
        \Gamma'\Ex{A}\seq C
    }
}
\]

This can be rewritten to the following.

\[
\infer[\mathcal{R}]{\Rx{\Delta}\seq C}{
    \mathcal{A}^? &
    \infer[\cut]{\Gamma'\Ex{\Delta}\seq C}{
        \Delta\seq A &
        \Gamma'\Ex{A}\seq C
    }
}
\]

Here, $\kappa$ is same in all cases, but $\delta$ decreases by one, so we apply the inductive hypothesis.

We need also to consider the applications of $(\to L)$ and $(\la L)$ where $A$ moves to the left assumption; up to symmetry this is the following.

\[
\infer[\cut]{\Rx{(\Px{\Delta},D\to E)}\seq C}{
    \Delta\seq A &
    \infer[\to L]{\Rx{(\Px{A},D\to E)}\seq C}{
        \Px{A}\seq D &
        \Rx{E}\seq C
    }
}
\]

This can be rewritten to the following.

\[
\infer[\to L]{\Rx{(\Px{\Delta},D\to E)}\seq C}{
    \infer[\cut]{\Px{\Delta}\seq D}{
        \Delta\seq A &
        \Px{A}\seq D
    } &
    \Rx{E}\seq C
}
\]

Once more, $\kappa$ is maintained while $\delta$ decreases by one.

Up to symmetry, we have the following rewrite for $(\tensor R)$.

\[
\infer[\cut]{(\Rx{\Delta},\Pi)\seq B\tensor C}{
    \Delta\seq A &
    \infer[\tensor R]{(\Rx{A},\Pi)\seq B\tensor C}{
        \Rx{A}\seq B &
        \Pi\seq C
    }
}
\quad\rwto
\]
\[
\infer[\tensor R]{(\Rx{\Delta},\Pi)\seq B\tensor C}{
    \infer[\cut]{\Rx{\Delta}\seq B}{
        \Delta\seq A &
        \Rx{A}\seq B
    } &
    \Pi\seq C
}
\]

Additionally, we have the following rewrite for $(\to R)$.

\[
\infer[\cut]{\Rx{\Delta}\seq B\to C}{
    \Delta\seq A &
    \infer[\to R]{\Rx{A}\seq B\to C}{
    (B,\Rx{A})\seq C
    }
}
\quad\rwto\quad
\infer[\to R]{\Rx{\Delta}\seq B\to C}{
\infer[\cut]{(B,\Rx{\Delta})\seq C}{
    \Delta\seq A &
    (B,\Rx{A})\seq C
}}
\]
\end{proof}

The case of $(\la R)$ is symmetric.  In all of these cases, $\kappa$ is maintained, but $\delta$ is decreased.

This exhausts all applications of rules to $\Rx{A}\seq C$ where $A$ is nonprincipal.

\subsubsection{Nonexponential Principal Pairs}
All remaining cases have $A$ principal in the final rule of the proofs of $\Delta\seq A$ and $\Rx{A}\seq C$.  As usual, we consider the top level connective of $A$ and show how to rewrite the proof by appealing to the inductive hypothesis.

In the case of tensor, we have the following

\[
\infer[\cut]{\Rx{(\Delta,\Pi)}\seq C}{
    \infer[\tensor R]{(\Delta,\Pi)\seq A\tensor B}{
        \Delta\seq A &
        \Pi\seq B
    } &
    \infer[\tensor L]{\Rx{A\tensor B}\seq C}{
    \Rx{(A,B)}\seq C
    }
}
\quad\rwto
\]
\[
\infer[\cut]{\Rx{(\Delta,\Pi)}\seq C}{
    \Delta\seq A &
    \infer[\cut]{\Rx{(A,\Pi)}\seq C}{
        \Pi\seq B &
        \Rx{(A,B)}\seq C
    }
}
\]

Both cut formulas have lower complexity, so we appeal twice to the inductive hypothesis.

Next we consider if the top level connective is $\to$.  Here we have the following rewrite.

\[
\infer[\cut]{\Rx{(\Delta,\Pi)}\seq C}{
    \infer[\to R]{\Pi\seq A\to B}{
    (A,\Pi)\seq B
    } &
    \infer[\to L]{\Rx{(\Delta,A\to B)}\seq C}{
        \Delta\seq A &
        \Rx{B}\seq C
    }
}
\quad\rwto
\]
\[
\infer[\cut]{\Rx{(\Delta,\Pi)}\seq C}{
    \Delta\seq A &
    \infer[\cut]{\Rx{(A,\Pi)}\seq C}{
        (A,\Pi)\seq B &
        \Rx{B}\seq C
    }
}
\]

Again, both applications of cut have smaller $\kappa$, so we appeal twice to the inductive hypothesis.

For the constant $\one$ we have

\[
\infer[\cut]{\Rx{}\seq C}{
    \infer[\one R]{\seq\one}{} &
    \infer[\one L]{\Rx{\one}\seq C}{
        \Rx{}\seq C
    }
}
\]

\noindent and the top proof of $\Rx{}\seq C$ suffices, removing the $(\cut)$.

\subsubsection{Exponential}
Lastly, we need only consider if the cut formula is of the form $!A$ and is active in the final rules of both proofs.  The only rules in which a banged succedent is active are $(\init)$, $(!R)$, $(!R\4)$, $(!R\K)$, and $(!R\K\4)$; for convenience call this class of rules, excluding $(\init)$, $\mathcal{RI}$.  We exclude $(\init)$ because we have already considered the case where the first rule above the left proof is $(\init)$.  The subset of the rules in $\mathcal{RI}$ that are in $\SMMacLL$ depends on the labels in $M$.

Following the deep cut elimination of Bra\"{u}ner and dePaiva \cite{brauneretal1996,brauneretal1998}, applied to noncommutative logics by Kanovich et. al. \cite{kanovich2017,kanovichetal2021}, the track all of the occurrences of $!A$ in the proof of $\Rx{!A}\seq C$ upwards to the introduction of their `bang', i.e. $(!)$.  The rules that can introduce bang on an antecedent are $(\init)$, $(\W)$, $(!L)$, $(!R)$, $(!R\K)$, and $(!R\K\4)$; for convenience label this class of rules $\mathcal{LI}$.

Hence, we are in the following scenario.

\[
\infer[\cut]{\Rx{!\Delta}\seq C}{
    \infer[\mathcal{RI}]{!\Delta\seq !A}{
    \Delta'\seq A
    } &
    \deduce{\Rx{!A}\seq C}{
        \deduce{\vdots}{
            \mathcal{LI}\text{ applications}
        }
    }
}
\]

Because of $(\C)$ or $(\C\K)$, there may be multiple applications of $\mathcal{LI}$ rules.  Note that these applications may be on the same or different branches of a proof.

Additionally, the form of $\Delta$ and $\Delta'$ both depend on the $\mathcal{RI}$ rule being applied, which itself depends on the labels in $M$ for the considered $\SMMacLL$.  For example, if $k\not\in M$, then $\Delta$ is one formula.  If the rule is $(!R\K)$ then $\Delta'\equiv\Delta$, but if it is $(!R\K4)$, then $\Delta'\equiv !^*\Delta$.

We modify the proof of $\Rx{!A}\seq C$ by replacing all occurrences of $!A$ with $!\Delta$, and then cutting in $A$ and/or adding certain rules elsewhere as necessary at each $\mathcal{LI}$ application.  This proof will only have cuts of lower complexity, and we then appeal to the inductive hypothesis several times.

It is not immediate for this replacement of a formula $!A$ by a structure $!\Delta$ to be valid generally.  If $k\not\in M$, then $\Delta$ must be a single formula, and we are fine.  If $k\in M$, then we use the fact that $(\C\K)$ allows for the contraction of entire substructures rather than only individual formulas, and the replacement is valid.

Consider each $\mathcal{LI}$ rule individually.

If the rule is $(\init)$, then in the modified proof, the sequent $!\Delta\seq !A$ appears, and we replace its proof with the proof ending in the $\mathcal{RI}$ application.  Symbolically,

\[
\infer[\init]{!A\seq !A}
\quad\rwto\quad
\infer[\mathcal{RI}]{!\Delta\seq !A}{
\infer{\Delta'\seq A}{\vdots}
}
\]

If the rule is $(\W)$, then $w\in S$, and we instead apply $(\W)$ individually to each formula in $!\Delta$.

\[
\infer[\W]{\Px{!A}\seq D}{
\Px{}\seq D
}
\quad\rwto\quad
\infer=[\W]{\Px{!\Delta}\seq D}{
\Px{}\seq D
}
\]

Depending on the $\mathcal{RI}$ rule applied, $\Delta'\equiv\Delta$ or $\Delta'\equiv !^*\Delta$.  In either case, we can achieve $\Delta'$ from $!\Delta$ by some, potentially zero, number of applications of $(!L)$. Therefore, for applications of $(!L)$ to $!A$, we can make the following replacement.

\[
\infer[!L]{\Px{!A}\seq D}{
\Px{A}\seq D
}
\quad\rwto\quad
\infer[\cut]{\Px{!\Delta}\seq D}{
    \infer=[!L]{!\Delta\seq A}{
        \infer{\Delta'\seq A}{\vdots}
    } &
    \Px{A}\seq C
}
\]

If the $\mathcal{LI}$ application is $(!R)$, then by the definition of $\SMMacLL$, we can conclude that the $\mathcal{RI}$ rule must be $(!R)$ or $(!R\4)$, and in particular $\Delta\equiv B$ for some $B$.  In both cases we have the following respective replacements.

\[
\infer[!R]{!A\seq !D}{
    A\seq D
}
\quad\rwto\quad
\infer[!R]{!B\seq !D}{
\infer[\cut]{B\seq D}{
    \infer{B\seq A}{\vdots} &
    A\seq D
}}
\]
\[
\infer[!R]{!A\seq !D}{
    A\seq D
}
\quad\rwto\quad
\infer[!R\4]{!B\seq !D}{
\infer[\cut]{!B\seq D}{
    \infer{!B\seq A}{\vdots} &
    A\seq D
}}
\]

The case of $(!R\K)$ is much the same.  The $\mathcal{RI}$ rule must have been $(!R\K)$.  We have the following rewrite.

\[
\infer[!R\K]{!\Px{!A}\seq !D}{
\Px{A}\seq D
}
\quad\rwto\quad
\infer[!R\K]{!\Px{!\Delta}\seq !D}{
\infer[\cut]{\Px{\Delta}\seq D}{
    \infer{\Delta\seq A}{\vdots} &
    \Px{A}\seq D
}
}
\]

Lastly, we need to consider applications of $(!R\K\4)$ where the bang on $!A$ is introduced.  Here, the $\mathcal{RI}$ rule must have also been $(!R\K\4)$.  In this case we have the following rewrite.

\[
\infer[!R\K]{!\Px{!A}\seq !D}{
!^*\Px{A}\seq D
}
\quad\rwto\quad
\infer[!R\K]{!\Px{!\Delta}\seq !D}{
\infer[\cut]{!^*\Px{!^*\Delta}\seq D}{
    \infer{!^*\Delta\seq A}{\vdots} &
    !^*\Px{A}\seq D
}
}
\]

This exhausts all cases, proving the lemma.  As a consequence, we have the following.

\begin{theorem}[Cut admissibility]
If a sequent is provable in $!^S_M\MacLL+(\cut)$, then it has a proof in $!^S_M\MacLL$.
\end{theorem}

\section{Light Systems}
We first consider systems that do not admit the axiom $\K$.  The modal terminology would be non-normal; we refer to them as light as in light linear logic \cite{girard1998}.  Fix some $S\subseteq\{c,w\}$ and $M\subseteq\{t,4\}$ that we will refer to throughout this section; in particular $k\not\in M$.

\subsection{Axiomatizability}
Our strategy is to reduce provability in $\SMMacLL$ to consequence in $\MacLL$.  Informally, we will show that finitely many nonlogical axioms are necessary to encode the behaviour of the exponential in $\SMMacLL$.

By Buli\'nska's work \cite{bulinska2009}, which is closely related to work of Buskowski \cite{buszkowski2005}, we have that the consequence relation is decidable in polynomial time, and that any fixed finite set of axioms yields context free languages when viewed as a categorial grammar.

Let $\quot{\cdot}:!\F\to\mathbb{N}$ be an injective efficiently computable encoding of multiplicative exponential formulas into natural numbers.

We map $!\F$ into $\F$ in the following way:

\begin{align*}
    \tau(p_i) & := p_{2i+1} \\
    \tau(A\tensor B) & := \tau(A)\tensor\tau(B) \\
    \tau(A\to B) & := \tau(A)\to\tau(B) \\
    \tau(B\la A) & := \tau(B)\la\tau(A) \\
    \tau(!A) & := p_{2\quot{A}} \\
\end{align*}

We extend this to structures in the obvious way; that is,

\[
\tau((\Gamma,\Delta)):=(\tau(\Gamma),\tau(\Delta))
\]

We now find an appropriate set of axioms $\Phi$ such that a sequent $\Gamma\seq C$ has

\[
\SMMacLL\vdash\Gamma\seq C
\quad\iff\quad
\MacLL+\Phi\vdash\tau(\Gamma)\seq\tau(C) 
\]

This $\Phi$ will need to depend on the sequent $\Gamma\seq C$, must be small, and in particular must be finite.

We define some technical functions to find the appropriate $\Phi$.

\begin{definition}
Let $\y$ be the \em{yield} of a structure; i.e. $\y(\Gamma)$ is the set of formulas in $\Gamma$ where

\begin{align*}
    \y((\Gamma,\Delta)) & := \y(\Gamma)\cup\y(\Delta) \\
    \y(A) & := \{A\}
\end{align*}
\end{definition}

\begin{definition}
Let $\sbf(A)$ be the subformulas of $A$, and for a set of formulas $S$, let $\sbf(S):=\bigcup_{A\in S}\sbf(A)$.
\end{definition}

\begin{definition}
Let $X\subseteq !\F$ be a set of formulas.  Define $\Phi^X_n$ to be the following set of sequents corresponding to provable formulas in $\SMMacLL$.

\begin{align*}
    \Phi^X_0 := 
    &\{p_{2\quot A}\seq \tau(A):!A\in X\} &\text{ if } t\in M\\
    \cup &\{p_{2\quot A}\seq p_{2\quot A}\tensor p_{2\quot A}:!A\in X\} &\text{ if } c\in S\\
    \cup &\{p_{2\quot A}\seq \one:!A\in X\} &\text{ if } w\in S\\
    \Phi^X_{n+1} := &\Phi^X_n\cup\{p_{2\quot A}\seq p_{2\quot B}:!A,!B\in X &\text{ and }\\
    &\MacLL+\Phi^X_n\vdash \tau(A)\seq \tau(B) &\text{ if } 4\not\in M\text{ or }\\
    &\MacLL+\Phi^X_n\vdash p_{2\quot A}\seq \tau(B) &\text{ if } 4\in M\}\\
    \Phi^X_{\infty}:=&\bigcup_n \Phi^X_n
\end{align*}
\end{definition}

Note that $\Phi^X_{\infty}$ is finite if $X$ is finite.

\begin{theorem}
Let $\Gamma\seq C$ be a sequent and $X$ a set of formulas.  If $\MacLL+\Phi^X_{\infty}$ proves the translation $\tau(\Gamma)\seq\tau(C)$, then $\SMMacLL$ proves $\Gamma\seq C$.
\end{theorem}
\begin{proof}
Since any proof uses finitely many axioms, it is sufficient to prove the claim for all $\Phi^X_n$, and we do this by induction on $n$.

We start with $\Phi^X_0$.  By \cite{bulinska2009}, we may restrict our attention to $\MacLL+\Phi^X_0$ proofs containing only propositional variables occurring in the endsequent $\Gamma\seq C$; this is a weakened form of the subformula property.

We proceed by a nested induction on such proofs in $\MacLL+\Phi^X_0$, translating the final rule.  The most interesting cases are $(\cut)$ and the nonlogical axioms.

If the last rule is $(\init)$, since $\tau$ is injective the claim follows trivially.

The propositional rules translate directly.  For example,

\[
\infer[\to L]{\tau(\Gamma)\Ex{(\tau(\Delta),\tau(A)\to\tau(B))}\seq C}{
    \tau(\Delta)\seq \tau(A) &
    \tau(\Gamma)\Ex{\tau(B)}\seq\tau(C)
}
\quad\rwto\quad
\infer[\to L]{\Rx{(\Delta,A\to B)}\seq C}{
    \Delta\seq A &
    \Rx{B}\seq C
}
\]

For a translation of $(\cut)$, it is a priori possible that the cut formula may not be the translation of a formula.  The general form is

\[
\infer[\cut]{\tau(\Gamma)\Ex{\tau(\Delta)}\seq \tau(C)}{
    \tau(\Delta)\seq A &
    \tau(\Gamma)\Ex{A}\seq \tau(C)
}
\]

In other words, if $\quot{\cdot}$ is not surjective, $\tau$ need not be surjective.  More specifically, some (even index) propositional variables may not be in the image of $\tau$.  However, by our earlier assumption, every propositional variable in the proof is in the endsequent $\tau(\Gamma)\seq\tau(C)$, which is in the image of $\tau$; consequently, every variable in the cut formula $A$ is the translation of an exponential multiplicative formula.  Thus, $A\equiv\tau(A')$ for some formula $A'$.  Then $(\cut)$ translates to $(\cut)$:

\[
\infer[\cut]{\tau(\Gamma)\Ex{\tau(\Delta)}\seq \tau(C)}{
    \tau(\Delta)\seq \tau(A') &
    \tau(\Gamma)\Ex{\tau(A')}\seq \tau(C)
}
\quad\rwto\quad
\infer[\cut]{\Rx{\Delta}\seq C}{
    \Delta\seq A' &
    \Rx{A'}\seq C
}
\]

Lastly, we check if the final (and thus only) rule is a nonlogical axiom of $\MacLL+\Phi_{!\F}$.  In other words, we check that these axioms hold.  We do this schemawise.

If $t\in M$:
\[
\infer[\Phi^X_0]{p_{2\quot A}\seq\tau(A)}{}
\quad\rwto\quad
\infer[\der]{!A\seq A}{
\infer[\init]{A\seq A}{}
}
\]

If $c\in S$:
\[
\infer[\Phi^X_0]{p_{2\quot A}\seq p_{2\quot A}\tensor p_{2\quot A}}{}
\quad\rwto\quad
\infer[\C]{!A\seq !A\tensor !A}{
\infer[\tensor R]{(!A,!A)\seq !A\tensor !A}{
    \infer=[\init]{!A\seq !A}{} &
    \infer=[\init]{!A\seq !A}{}
}}
\]

If $w\in S$:
\[
\infer[\Phi^X_0]{p_{2\quot A}\seq \one}{}
\quad\rwto\quad
\infer[\W]{!A\seq\one}{
\infer[\one R]{\seq\one}{}}
\]

This exhausts the rules, concluding the base case.

Now we assume the claim for $n$ to prove it for $n+1$, proceeding by nested induction on $\MacLL+\Phi^X_{n+1}$ proofs.  For all final rules also appearing in $\MacLL+\Phi^X_0$, we translate as above.  Thus, we need only consider nonlogical axioms of the form $p_{2\quot A}\seq p_{2\quot B}$.

If $4\not\in M$ and $p_{2\quot A}\seq p_{2\quot B}\in\Phi^X_{n+1}$, then we have $\MacLL+\Phi^X_n\vdash p_{2\quot A}\seq\tau(B)$.  By the inductive hypothesis, this tells us that $\SMMacLL\vdash !A\seq B$, so we translate as follows:

\[
\infer[\Phi^X_{n+1}]{p_{2\quot A}\seq p_{2\quot B}}{}
\quad\rwto\quad
\infer[!R\4]{!A\seq !B}{\infer{!A\seq B}{\vdots}}
\]

In the case $4\in M$, if $p_{2\quot A}\seq p_{2\quot B}\in\Phi^X_{n+1}$ then $\MacLL+\Phi^X_n\vdash\tau(A)\seq\tau(B)$, which by the inductive hypothesis tells us that $\SMMacLL\vdash A\seq B$, so we again translate using this implied $\SMMacLL$ proof:

\[
\infer[\Phi^X_{n+1}]{p_{2\quot A}\seq p_{2\quot B}}{}
\quad\rwto\quad
\infer[!R]{!A\seq !B}{\infer{A\seq B}{\vdots}}
\]

Thus for any $n$, the axioms of $\Phi^X_n$ are sound with respect to $\SMMacLL$.
\end{proof}

We now prove a converse result.

\begin{definition}
We define $\Phi^{\Gamma,C}_n$ to be $\Phi^X_n$ where $X$ is the set of all subformulas in the sequent $\Gamma\seq C$.  Formally,

\[
    \Phi^{\Gamma,C}_n:=\Phi^{\sbf(\y(\Gamma)\cup\{C\})}_n
\]
\end{definition}

Note that $\sbf(\y(\Gamma)\cup\{C\})$ is finite, and thus $\Phi_{\Gamma,C}$ is finite.  This is sufficient to prove everything provable in $\SMMacLL$.

\begin{theorem}
Let $\Gamma\seq C$ be provable in $\MacLL$, then $\tau(\Gamma)\seq \tau(C)$ is provable in $\MacLL+\Phi^{\Gamma,C}_{\infty}$. 
\end{theorem}
\begin{proof}
We prove this by induction on the height of minimal cutfree $\SMMacLL$ proofs.  We consider the last rule casewise.

If the last rule is an axiom, the result is immediate:

\[
\infer[\init]{A\seq A}{}
\quad\rwto\quad
\infer[\init]{\tau(A)\seq\tau(A)}{}
\]

If the last rule is a non-exponential propositional rule, we translate directly and appeal to the induction hypothesis; for example,

\[
\infer[\to L]{\Rx{(\Delta,A\to B)}\seq C}{
    \Delta\seq A &
    \Rx{B}\seq C
}
\quad\rwto\quad
\infer[\to L]{\tau(\Gamma)\Ex{(\tau(\Delta),\tau(A)\to\tau(B))}\seq C}{
    \tau(\Delta)\seq \tau(A) &
    \tau(\Gamma)\Ex{\tau(B)}\seq\tau(C)
}
\]


Most interesting is if an exponential rule is last.  If the last rule is $(!L)$, then $t\in M$ and we have

\[
\infer[!L]{\Rx{!A}\seq C}{\Rx{A}\seq C}
\quad\rwto\quad
\infer[\cut]{\tau(\Gamma)\Ex{p_{2\quot A}}\seq \tau(C)}{
    \infer[\Phi^{\Rx{!A},C}_{\infty}]{p_{2\quot A}\seq \tau(A)}{} &
    \tau(\Gamma)\Ex{\tau(A)}\seq \tau(C)
}
\]

\noindent since $!A$ is a subformula (and in fact a formula) in $\Rx{!A}\seq C$, and thus $p_{2\quot A}\seq \tau(A)\in\Phi^{\Rx{!A},C}_0\subseteq\Phi^{\Rx{!A},C}_{\infty}$.

The case of $(!R)$ relies more strongly on the construction of $\Phi_{\Gamma,C}$.

Say the last rule is

\[
\infer[!R]{!A\seq !B}{A\seq B}.
\]

Thus $4\not\in M$ and $\SMMacLL\vdash A\seq B$.  By the inductive hypothesis, this tells us that

\[
\MacLL+\Phi^{A,B}_{\infty}\vdash \tau(A)\seq\tau(B).\]

By expanding the axiom set, we also then have,

\[
\MacLL+\Phi^{!A,!B}_{\infty}\vdash \tau(A)\seq\tau(B).
\]

By finiteness, for some $n$ we then also have

\[
\MacLL+\Phi^{!A,!B}_n\vdash \tau(A)\seq\tau(B).
\]

\noindent which by construction implies

\[
p_{2\quot A}\seq p_{2\quot B}\in\Phi^{!A,!B}_{n+1}\subseteq\Phi^{!A,!B}_{\infty}.
\]

Therefore, in this case, the translation is provable by a nonlogical axiom.

\[
\infer[\Phi^{!A,!B}_{\infty}]{p_{2\quot A}\seq p_{2\quot B}}{}
\]

If the last rule application is

\[
\infer[!R4]{!A\seq !B}{!A\seq B}.
\]

\noindent then $4\in M$ and $\SMMacLL\vdash !A\seq B$.  We similarly have then that $p_{2\quot A}\seq p_{2\quot B}\in\Phi^{!A,!B}_{\infty}$, so again the inductive claim follows by a nonlogical axiom.

If the final rule is $(\C)$, then $c\in S$, and we can translate as follows.

\[
\infer[\C]{\Rx{!A}\seq C}{\Rx{(!A,!A)}\seq C}
\quad\rwto
\]
\[
\infer[\cut]{\tau(\Gamma)\Ex{p_{2\quot A}}\seq \tau(C)}{
    \infer[\Phi^{\Gamma,C}_{\infty}]{p_{2\quot A}\seq p_{2\quot A}\tensor p_{2\quot A}}{} &
    \infer[\tensor L]{\tau(\Gamma)\Ex{p_{2\quot A}\tensor p_{2\quot A}}\seq \tau(C)}{
        \tau(\Gamma)\Ex{(p_{2\quot A},p_{2\quot A})}\seq \tau(C)
    }
}
\]

\noindent using the fact that $!A$ is a subformula (and in fact a formula) in $\Rx{!A}$.

Lastly, if the final rule is $(\W)$, then $w\in S$, and we can carry out the following translation.

\[
\infer[\W]{\Rx{!A}\seq C}{\Rx{}\seq C}
\quad\rwto
\]
\[
\infer[\cut]{\tau(\Gamma)\Ex{p_{2\quot A}}\seq \tau(C)}{
    \infer[\Phi^{\Gamma,C}_{\infty}]{p_{2\quot A}\seq \one}{} &
    \infer[\one L]{\tau(\Gamma)\Ex{\one}\seq \tau(C)}{
        \tau(\Gamma)\Ex{}\seq \tau(C)
    }
}
\]

This exhausts all cases, proving the claim.
\end{proof}

We can put these results together to obtain a soundness an completeness theorem.

\begin{corollary}[Soundness and Completeness]
Let $\Gamma\seq C$ be a sequent, then $\SMMacLL$ proves $\Gamma\seq C$ if and only if $\MacLL+\Phi^{\Gamma,C}_{\infty}$ proves $\tau(\Gamma)\seq\tau(C)$.
\end{corollary}

\subsection{$\PTIME$ Decidability}

\begin{theorem}
If $S\subseteq\{c,w\}$ and $M\subseteq\{t,4\}$, then provability in $\SMMacLL$ is $\PTIME$.
\end{theorem}
\begin{proof}
Fix some $\Gamma\seq C$ with $N$ subformulas noting that the number of subformulas is linear in the complexity of the sequent.  We give an explicit polynomial time algorithm, starting with a construction of $\Phi^{\Gamma,C}_{\infty}$.

Note that $\Phi^{\Gamma,C}_0$ can be constructed in linear time.

There are at most $N^2$ axioms of the form $p_{2\quot A}\seq p_{2\quot B}$ where $!A$ and $!B$ are subformulas of formulas in $\Gamma\seq C$.  To construct $\Phi^{\Gamma,C}_{n+1}$ from $\Phi^{\Gamma,C}_n$, we will check membership of each $p_{2\quot A}\seq p_{2\quot B}$ individually.  By the result of Buli\'nska \cite{bulinska2009}, we can check provability in $\MacLL+\Phi^{\Gamma,C}_n$ in polynomial time, parametric on the axiom set and $|\Phi^{\Gamma,C}_n|\leq N^2+3N$.  Thus, construction of each successive $\Phi^{\Gamma,C}_n$ requires only polynomial time.

Since at most $N^2$ new axioms can be added, $\Phi^{\Gamma,C}_n$ stabilizes in at most $N^2$ stages.  Thus, running until stabilization takes polynomial time, and yields $\Phi^{\Gamma,C}_{\infty}$.

By soundness and completeness, $\SMMacLL\vdash\Gamma\seq C$ if and only if $\MacLL+\Phi^{\Gamma,C}_{\infty}\vdash\tau(\Gamma)\seq\tau(C)$, and again using the result of Buli\'nska \cite{bulinska2009}, we can check if $\MacLL+\Phi^{\Gamma,C}_{\infty}$ proves $\tau(\Gamma)\seq\tau(C)$ in polynomial time.
\end{proof}

\section{Context Freeness}
Appealing to the same paper of Buli\'nska \cite{bulinska2009}, we pin down exactly the recognizing strength of $\SMMacLL$ as a categorial grammar \cite{Lambek1958,moortgat1997categorial}.

\begin{theorem}
For $S\subseteq\{c,w\}$ and $M\subseteq\{t,4\}$, then $\SMMacLL$ recognizes exactly the context free languages as a categorial grammar.
\end{theorem}
\begin{proof}
By cut admissibility, $\MacLL$ is a conservative fragment of $\SMMacLL$.  Thus, by \cite{Kandulski1988TheEO}, we know that $\SMMacLL$ recognizes at least the context free languages.  We now show that it recognizes only the context free languages.

Fix some finite alphabet $\Sigma$, some formula $s\in!\F$, and some lexicon, i.e. $\Lex:\Sigma\to\Pfin(!\F)$.  Then define

\[
L:=\bigcup_{a\in\Sigma}\sbf(\Lex(a))\cup\sbf(s)
\]

\noindent to be the set of subformulas of all formulas appearing in the categorial grammar, and note that it is finite.  Therefore, $\Phi^L_{\infty}$ is also finite.

Further, define $\tau\Lex:\F\to\Pfin(\F)$ by $\tau\Lex(a):=\tau(\Lex(a))$.  We will consider this as a lexicon for $\MacLL+\Phi_{L}$.

Pick an arbitrary string $a_1\cdots a_n$.  We will show that it is recognized by $\SMMacLL$ with $\Lex$ and $s$ iff it is recognized by $\MacLL+\Phi^L_{\infty}$ with $\tau\Lex$ and $\tau(s)$.

Say $\SMMacLL$ recognizes $a_1\cdots a_n$, then there are $A_i\in\Lex(a_i)$ and a bracketing $\Gamma$ of $A_1,...,A_n$ such that $\SMMacLL\vdash\Gamma\seq s$.  Note that all subformulas of $\Gamma\seq s$ are in $L$, so by our soundness and completeness we have $\tau(\Gamma)\seq\tau(s)$.  Since $\tau(\Gamma)$ is a bracketing of $\tau(A_1),...\tau(A_n)$, and by definition we have $\tau(A_i)\in\tau\Lex(a_i)$, this tells us that $\MacLL+\Phi_{L}$ recognizes $a_1\cdots a_n$, proving one direction.

Now say that $\MacLL+\Phi^L_{\infty}$ recognizes $a_1\cdots a_n$.  By definition there are $\tau(A_i)\in\tau\Lex(a_i)$ (with $A_i\in\Lex(a_i)$) and a bracketing $\tau(\Gamma)$ of $\tau(A_1),...,\tau(A_n)$ such that $\MacLL+\Phi^L_{\infty}\vdash\tau(\Gamma)\seq\tau(s)$.  Again by soundness and completeness, $\SMMacLL\vdash\Gamma\seq s$.  Note that $\Gamma$ is a bracketing of $A_1,...,A_n$ and thus $\SMMacLL$ recognizes $a_1\cdots a_n$.

Thus, $\SMMacLL$ recognizes exactly the same strings as $\MacLL+\Phi^L_{\infty}$.  We have by Buli\'nska \cite{bulinska2009} that $\MacLL+\Phi^L_{\infty}$ recognizes only context free languages, so $\SMMacLL$ must also recognize only context free languages.
\end{proof}

\section{Related and Future Work}
This paper is purely multiplicative.  It is unclear how these results extent to systems with additives.  By work of Kuznetsov \cite{kuznetsov2013}, we know that in the associative noncommutative system, the addition of additives expands the recognizing power beyond context-free grammars.  If in addition we include the distribution axiom $A\wedge(B\vee C)\seq (A\vee B)\wedge (A\vee C)$, Buszkowski and Farulewski show that the categorial grammars remain context-free \cite{buszkowskiDFNL}.

In another direction, the exponential of linear logic is not canonical \cite{danos93kgc}.  The addition of multiple exponentials to a logical system \cite{miller09subexponentials} is incredibly flexible.  The arguments in this paper likely extend to related multimodal systems.

Pinning down the computational complexity and recognizing power of the systems of this paper in the presence of additives and subexponentials remains open.

After completion of this work, the author became aware of an abstract by Wang, Ge, and Lin \cite{wang2023}.  They prove $\PTIME$ decidability of a similar but different logic.  They consider a logic equivalent to $!^{c,w}_{k,t,4}\MacLL$, which is in particular not light.  Neither decidability follows directly from the other.

%
%
%

%
%
%
%
\bibliographystyle{splncs04}
\bibliography{biblio}

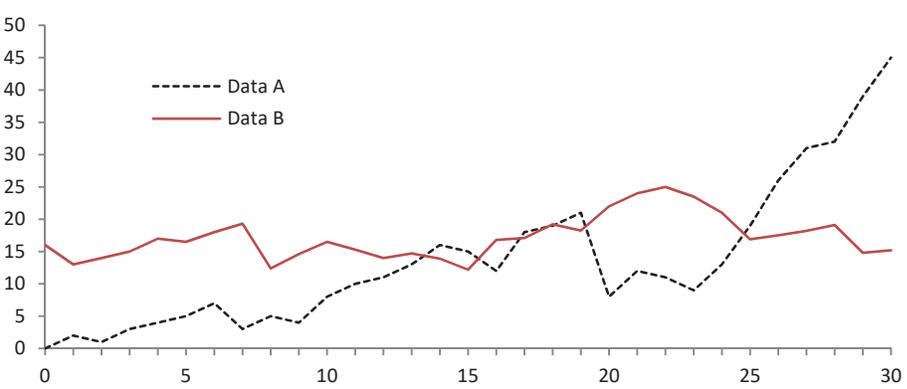
\end{document}